\documentclass[11pt,fleqn]{article}
\usepackage{amssymb,latexsym,amsmath,amsfonts,graphicx, ebezier}

    \advance\textheight by \topskip

    \numberwithin{equation}{section}

    \def\tr{{\rm tr \,}}
    \def\Re{{\rm Re \,}}
    \def\Im{{\rm Im \,}}

    \def\bigO{{\cal O}}
    
    \def\Res{{\rm Res}}

    \def\P2n{{\rm P}_{{\rm II}}^{(n)}}

    \newtheorem{theorem}{Theorem}[section]
    \newtheorem{lemma}[theorem]{Lemma}
    \newtheorem{corollary}[theorem]{Corollary}
    \newtheorem{proposition}[theorem]{Proposition}
    
    \newtheorem{Definition}[theorem]{Definition}
    
    \newtheorem{Remark}[theorem]{Remark}
    \newenvironment{remark}{\begin{Remark}\rm}{\end{Remark}}
    \newtheorem{Example}[theorem]{Example}
    
    \newtheorem{Assumptions}[theorem]{Assumptions}

    \newcommand{\e}{\epsilon}
\newcommand{\lb}{\lambda}

    \newenvironment{proof}%
    {\rm \trivlist \item[\hskip \labelsep{\bf Proof. }]}%
    {\hspace*{\fill}$\Box$\endtrivlist}
    {\rm \trivlist \item[\hskip \labelsep{\bf Proof}]}%
    {\hspace*{\fill}$\Box$\endtrivlist}

    \newcommand{\sgn}{{\operatorname{sgn}}}
    
    \DeclareMathOperator*{\Tr}{Tr}

\begin{document}
\title{Pole-free solutions of the first Painlev\'e hierarchy and non-generic critical behavior for the KdV equation}
\author{Tom Claeys$^1$}

\maketitle
\footnotetext[1]{Universit\'e Catholique de Louvain, Chemin du cyclotron 2, B-1348 Louvain-La-Neuve, Belgium,\\
E-mail: tom.claeys@uclouvain.be, Tel.: +32 10 47 31 89, Fax: +32 10 47 25 30}

\begin{center}
{\em Dedicated to Boris Dubrovin on the occasion of his sixtieth birthday}
\end{center}
\begin{abstract}
We establish the existence of real pole-free solutions to all even
members of the Painlev\'e I hierarchy. We also obtain asymptotics
for those solutions and describe their relevance in the description
of critical asymptotic behavior of solutions to the KdV equation in
the small dispersion limit. This was understood in the case of a
generic critical point, and we generalize it here to the case of
non-generic critical points.
\end{abstract}

\section{Introduction and statement of results}
 We first introduce the
Painlev\'e I hierarchy, which has the Painlev\'e I equation
\begin{equation}
q_{ss}=s+6q^2
\end{equation}
as its first member. The $m$-th equation of the hierarchy is of order $2m$ and is defined recursively. It has the form
\begin{equation}\label{PIm}
s+\mathcal L_m(q)+\sum_{j=1}^{m-1}t_j\mathcal L_{j-1}(q)=0,\qquad
t_1,\ldots, t_{m-1}\in\mathbb R,
\end{equation}
where $\mathcal L_m$ is the Lenard-Magri recursion operator defined by
\begin{align}
&\label{L0}\mathcal L_0(q)=-4 q,\\
&\label{recursion Lm}\frac{d}{ds}\mathcal
L_{k+1}(q)=\left(\frac{1}{4}\frac{d^3}{ds^3}-2q\frac{d}{ds}-q_s\right)\mathcal
L_k(q),\qquad \mbox{ for $k=0, \ldots, m-1$.}
\end{align}
The constants of integration in (\ref{recursion Lm}) are fixed by
the requirements $\mathcal L_1(0)=\cdots =\mathcal L_{m}(0)=0$. One
could also add a term $t_{m}\mathcal L_{m-1}$ in (\ref{PIm}), but
this term cancels after a simple transformation of the other
variables $q,s,t_1, \ldots, t_{m-1}$. The first equations in the
hierarchy are given by (up to multiplication by a non-zero constant)
\begin{align*}
&m=0: && s-4q=0,\\
&m=1: && q_{ss}=s+6q^2,\\
&m=2: && q_{ssss}=4s-40q^3+10q_s^2+20qq_{ss}-16t_1q,\\
&m=3: &&
q^{(6)}=16s+28qq_{ssss}+56q_sq_{sss}+42q_{ss}^2-280(q^2q_{ss}+qq_s^2-q^4)\\
&\ &&\qquad\qquad\qquad\qquad\qquad+16t_{2}(6q^2-q_{ss})-64t_1q,\\
&m=4: &&q^{(8)}=64s+36qq^{(6)}+108q_sq^{(5)}+228q_{ss}q_{ssss}-504q^2q_{ssss}\\
&\ &&\qquad\qquad\qquad+138q_{sss}^2-1512qq_{ss}^2-1848q_s^2q_{ss}-2016 qq_sq_{sss}\\
&\ &&\qquad\qquad\qquad\qquad\qquad\qquad -2016 q^5 +3360q^3q_{ss}+5040q^2q_s^2\\
&\ &&+16t_3(-40q^3+10q_s^2+20qq_{ss}-q^{(4)})\\
&\ &&\qquad\qquad\qquad\qquad\qquad+64t_{2}(6q^2-q_{ss})-256t_1q.
\end{align*}
In the above equations we have written $q^{(j)}$ for the $j$-th
derivative of $q$ with respect to $s$. We will call equation
(\ref{PIm}) the $P_{\rm I}^{m}$ equation, and we refer to \cite{Kud,
MuganJrad, Shimomura, Takasaki, GoPi} for more information about the
first Painlev\'e hierarchy. Given $t_1, \ldots, t_{m-1}$, solutions
to these equations are meromorphic functions in the complex
$s$-plane, with in general an infinite number of poles
\cite{Shimomura}. 
Heuristic arguments supporting the existence of real pole-free solutions to the
    even members of the hierarchy were already given in \cite{BMP} in the case where
    $t_1=\ldots=t_{m-1}=0$. Their asymptotic behavior was discussed in \cite{Moore}.
A particular solution to the second member of the hierarchy appeared in \cite{Suleimanov1, Suleimanov2} in relation with the Gurevich-Pitaevskii special solutions to 
the KdV equation \cite{GP, Potemin}. This solution corresponds to the solution studied in \cite{BMP} for $t_1=0$, and has applications in the study of ideal incompressible liquids \cite{KS1, KS2} and quantum gravity \cite{DSS}, see also \cite{GST}.

After the rescalings
\[U=-60^{2/7}\cdot q,\qquad  X=60^{-1/7}\cdot s,\qquad  T=-4\cdot 60^{-3/7}\cdot t_1,\]
the second member of the hierarchy $P_{\rm I}^2$ becomes
\begin{equation}\label{PI2}
        X=TU-\left(\frac{1}{6}U^3+\frac{1}{24}(U_X^2+2UU_{XX})
            +\frac{1}{240}U_{XXXX}\right).
    \end{equation}
    This equation was studied by Dubrovin in \cite{Dubrovin},
    where he conjectured the existence and uniqueness of a real
    solution without poles for real values of $X$ and $T$. The
    existence of such a solution was proved in \cite{CV1}
    together with the asymptotic behavior $U(X,T)\sim \mp
    (6|X|)^{1/3}$ as $X\to\pm\infty$. 
   In addition it is known that
    $U(X,T)$ is also a solution to the KdV equation $U_T+UU_X+\frac{1}{12}U_{XXX}=0$
    \cite{GP, Potemin, Suleimanov1, Suleimanov2}.

\medskip

As a first result in this paper, we will prove the existence of real pole-free solutions for all even members of the hierarchy, we
will obtain asymptotics for them, and show that they follow, as functions of the time variables $t_1, \ldots, t_{m-1}$, the time
flows of the KdV hierarchy.
\begin{theorem}\label{theorem: PIm}
Let $m$ be an even positive integer. There exists a solution
$q=q(s,t_1, \ldots, t_{m-1})$ to equation (\ref{PIm}) which has the
properties
\begin{itemize}
\item[(i)] $q$ is real and has no poles for real values of $s,t_1, \ldots,
t_{m-1}$; for $s, t_1, \ldots, t_{m-1}$ in a sufficiently small
neighborhood of the real line, $q$ depends analytically on each of
its variables,
\item[(ii)] $q$ satisfies the PDE
\begin{equation}\label{KdVm}
q_{t_k}+\frac{1}{2k+1}\frac{d}{ds}\mathcal L_k=0,\qquad\mbox{ for
$k=1, \ldots, m-1$,}
\end{equation}
which is (up to re-scaling) the $k$-th equation in the KdV
hierarchy,
\item[(iii)] for $t_1, \ldots, t_{m-1}=0$, $q$ has the asymptotic behavior
\begin{equation}\label{as q}
q(s,0, \ldots,
0)=c|s|^{\frac{1}{m+1}}+\bigO(|s|^{-\frac{m}{m+1}}),\qquad\mbox{ as
$s\to\pm\infty$,}
\end{equation}
with
\begin{equation}
c=\frac{\sgn(s)}{2}\left(\frac{2^{m-1}(m+1)!}{(2m+1)!!}\right)^{\frac{1}{m+1}}.
\end{equation}
\end{itemize}
\end{theorem}
\begin{remark}
The above results are not true for $m=1$, since it is known
 that there do not exist real pole-free solutions
to the Painlev\'e I equation \cite{Boutroux}. Also for $m>1$ odd,
we do not expect that our results can be generalized. Parts (i) and
(iii) are the essential parts of the theorem, part (ii) will follow
from rather standard arguments that express the relation between the
Painlev\'e I hierarchy and the KdV hierarchy.
\end{remark}
\begin{remark}For $t_1,\ldots,t_{m-1}\in\mathbb R$ fixed, our asymptotic analysis can be
generalized to obtain the asymptotics
\begin{equation}\label{as q2}
q(s,t_1, \ldots,
t_{m-1})=c|s|^{\frac{1}{m+1}}+\bigO(|s|^{-\frac{1}{m+1}}),\qquad\mbox{
as $s\to\pm\infty$,}
\end{equation}
see Remark \ref{remark tj} below. Note that the error term is weaker in
the case of non-zero $t_j$'s than in (\ref{as q}).
\end{remark}
\begin{remark}
The leading order of the asymptotic behavior of $q$ as
$s\to\pm\infty$ for fixed $t_1, \ldots, t_{m-1}$ is relatively
simple, and can be formally obtained when neglecting all derivatives
of $q$ in the equations of the hierarchy. It was proved in \cite{DaiZhang} that solutions to the 
Painlev\'e I hierarchy with this asymptotic behavior exist as $x\to \infty$ in a sector containing the positive real line. If one takes double
scaling limits where $t_1, \ldots, t_{m-1}$ tend to infinity
simultaneously with $s$, the situation becomes more complicated.
Then the type of asymptotics will depend on the precise scaling of
all variables. For example it can be expected that the asymptotics
for $q$ can be expressed in terms of elliptic $\theta$-functions in
some regions and, for $m>2$, in terms of hyperelliptic
$\theta$-functions in other regions. For critical scalings of the
variables, one can even expect asymptotics for $q=q_m$ in terms of
the pole-free solutions $q_2, q_4, \ldots, q_{m-2}$ of the lower
order equations in the Painlev\'e I hierarchy, and in terms of
certain solutions to the Painlev\'e II hierarchy, see the discussion
in \cite{KS1} and \cite{C}.
\end{remark}

\subsection{Critical behavior for KdV solutions in the small dispersion limit}

Let us consider the KdV equation
\begin{equation}\label{KdV}u_t+6uu_x+\e^2u_{xxx}=0,\qquad \e>0.\end{equation}
For small $\e$, this is an example of a Hamiltonian perturbation of
the Hopf equation $u_t+6uu_x=0$. Solutions to the Hopf equation
exist only for small times, and develop a point of gradient
catastrophe after a certain time. Indeed, consider for example
initial data $u_0$ which are negative, smooth, tending to $0$
rapidly at $\pm\infty$ and with a single local minimum. Then the
method of characteristics describes the solution in terms of the
initial data:
\[ u(x,t)=u_0(\xi),\qquad x=6tu_0(\xi)+\xi,
\]
and the slope becomes infinite at the critical time
\[
t_c=\dfrac{1}{\max_{\xi\in\mathbb{R}}[-6u'_0(\xi)]}.
\]
The point $x_c$ and time $t_c$ where $u_x$ blows up, and the value
$u_c=u(x_c,t_c)$ are also determined by the equations
\begin{align}
\label{xc1}
&F(u;x,t):=-x+6ut+f_L(u)=0,\\
\label{tc1}
&F'(u;x,t)=6t+f_L'(u)=0,\\
\label{uc1} &F''(u;x,t)=f_L''(u)=0,
\end{align}
where $f_L$ is the inverse of the decreasing part of the initial
data $u_0$. For generic initial data we have $f_L'''(u_c)\neq 0$,
which means that the Hopf solution behaves locally as
\begin{equation}u(x,t_c)=u_c-c(x-x_c)^{1/3}+\bigO(x-x_c),\qquad\mbox{ as $x\to
x_c$}.\end{equation} For non-generic initial data however, it can
happen that $f_L'''(u_c)=0$ and that
\begin{equation}u(x,t_c)= u_c-c(x-x_c)^{\frac{1}{m+1}}+\bigO(x-x_c),\qquad\mbox{ as $x\to
x_c$},\end{equation} for any even value of $m$. This is the case if
the initial data are such that
\begin{equation}\label{non-generic}
f_L^{(3)}(u_c)=f_L^{(4)}(u_c)=\ldots = f_L^{(m)}(u_c)=0, \qquad
f_L^{(m+1)}(u_c)\neq 0.
\end{equation}

 It was conjectured
by Dubrovin \cite{Dubrovin, Dubrovin2, Dubrovin3, Dubrovin4, DGK}
that the behavior of generic solutions to any Hamiltonian
perturbation of a hyperbolic equation near the critical point of the
unperturbed equation is described universally in terms of the
pole-free solution to the second member (of order $4$) of the
Painlev\'e I hierarchy: there should be an expansion of the form
\begin{equation}
\label{univer} u(x,t,\e)= u_c +a_1\e^{2/7}U \left( a_2\e^{-6/7}(x-
x_c-a_3(t-t_c)), a_4\e^{-4/7}(t-t_c)\right) +\bigO\left(
\epsilon^{4/7}\right),
\end{equation}
in a double scaling limit where $\e\to 0$ and simultaneously $x\to
x_c$, $t\to t_c$ in such a way that $\e^{-6/7}(x- x_c-a_3(t-t_c))$
and $\e^{-4/7}(t-t_c)$ tend to real constants, and where $U$ is the
pole-free solution to equation (\ref{PI2}). The values of $a_1,
\ldots, a_4$ depend on the equation and the initial data, but not on
$x,t,\e$.
 In the case of the KdV equation, this was proved afterwards \cite{CG} for analytic
negative initial data with sufficient decay at $\pm\infty$ and with
a single local minimum, under the condition that the gradient
catastrophe for the Hopf equation is generic, i.e.\
$f_L'''(u;x_c,t_c)\neq 0$. If this condition is not satisfied, we
will prove here that the KdV solution is no longer described in
terms of the pole-free solution to the $P_{\rm I}^2$ equation, but
in terms of a solution to a higher order equation in the Painlev\'e
I hierarchy satisfying the properties given in Theorem \ref{theorem:
PIm}. The order of the equation is determined by the number of
vanishing derivatives of $f_L$: if we have (\ref{non-generic}) for
$m\in 2\mathbb N$, the pole-free solution to the $P_{\rm I}^{m}$
equation of order $2m$ will appear.

\medskip

We consider initial data $u_0(x)$ in the class of negative functions
with only one local minimum, and such that $u_0$ can be extended to
an analytic function $u_0(z)$ in a region of the form
\[ \mathcal{S}=\{z\in\mathbb{C}: |\Im z|<\tan\theta|\Re
z|\}\cup\{z\in\mathbb C: |\Im z|<\sigma\}
\]
for some $0<\theta<\pi/2$ and $\sigma>0$. In addition we need
sufficient decay at infinity,
\begin{equation}
u_0(x)=\bigO\left(\frac{1}{|x|^{3+s}}\right), \;\;s>0, \quad x\in
\mathcal S,\quad x\to\infty.
\end{equation}
The local minimum is localized at a point $x_M$ and we assume that
$u_0''(x_M)\neq 0$ and $u_0(x_M)=-1$.

\begin{theorem}\label{theorem: KdV}
Let $u_0(x)$ be initial data for the Cauchy problem of the KdV
equation satisfying the conditions described above, and assume that
we have (\ref{non-generic}) for $m\in 2\mathbb N$. Write
$u_c=u(x_c,t_c)$ for the Hopf solution at the point $x_c$ and time
$t_c$ of gradient catastrophe of the  Hopf equation. We take a
double scaling limit where we let $\epsilon\to 0$ and at the same
time $x\to x_c$ and $t\to t_c$ in such a way that, for some
$\tau_0,\tau_1\in\mathbb R$,
\begin{equation}
\lim\frac{x-x_c-6u_c(t-t_c)}{k^{1/2}\e^{\frac{2m+2}{2m+3}}}=\tau_0,
\qquad -\lim\frac{3(t-t_c)}{k^{3/2}\e^{\frac{2m}{2m+3}}}= \tau_1,
\end{equation}
where
\begin{equation}\label{k}
k=\left(-\frac{2^{m-1}}{(2m+1)!!}f_L^{(m+1)}(u_c)\right)^{\frac{2}{2m+3}}>0.
\end{equation} In this double scaling limit the solution $u(x,t,\epsilon)$ of
the KdV equation (\ref{KdV}) admits the asymptotic expansion
\begin{equation}
\label{expansionu} u(x,t,\e)=u_c -\frac{2}{k}\e^{\frac{2}{2m+3}} q
\left(\e^{-\frac{2m+2}{2m+3}}\tau_0(x,t,\e),\e^{-\frac{2m}{2m+3}}\tau_1(t,\e),0,0,\ldots,
0\right) +\bigO( \epsilon^{\frac{4}{2m+3}}),
\end{equation}
where
\begin{equation}
\tau_0(x,t,\e)=\frac{x-x_c-6u_c(t-t_c)}{k^{1/2}},\qquad
\tau_1(t,\e)=-\frac{3(t-t_c)}{k^{3/2}}.
\end{equation}
Here $q(s,t_1, t_2, \ldots, t_{m-1})$ is a solution to equation
(\ref{PIm}) which has properties (i)-(ii)-(iii) given in Theorem
\ref{theorem: PIm}.
\end{theorem}
\begin{remark}
In the case of a generic critical point where $m=2$, Theorem
\ref{theorem: KdV} describes the main theorem proved in \cite{CG}.
The result is new for $m>2$.
\end{remark}
\begin{remark}
It is likely that the above result holds also for other equations in
the universality class of Dubrovin \cite{Dubrovin}, which contains
among others the Camassa-Holm equation, the de-focusing nonlinear
Schr\"odinger equation, and the KdV hierarchy.
\end{remark}
\subsubsection{Correlation kernels in critical unitary random matrix ensembles}
The first conjecture about the existence of real pole-free solutions
to the $P_I^m$ equations for even $m$ was posed in random matrix
theory \cite{BB, BMP} in the case where $t_1=\ldots=t_{m-1}=0$.
However, it leads no doubt that also the general case with non-zero
$t_j$'s is of interest when studying critical unitary random matrix
ensembles. Consider the space of Hermitian $n\times n$ matrices with
a probability distribution of the form
\begin{equation}\label{RMT}\frac{1}{Z_{n}}\exp(-n\,\tr V(M))dM,\end{equation} where $V$ is a scalar
real analytic function with sufficient growth at $\pm\infty$, for
example a polynomial of even degree with positive leading
coefficient. The limiting mean eigenvalue distribution for random
matrices in such an ensemble has the form \cite{DKM}
\begin{equation*}
d\mu_V(x)=\psi_V(x)dx,\qquad
\psi_V(x)=\prod_{j=1}^k\sqrt{(b_j-x)(x-a_j)}h(x),\quad\mbox{for
$x\in\cup_{j=1}^k[a_j,b_j]$,}
\end{equation*}
where $h$ is a real analytic function. The number of intervals in
the spectrum, the endpoints $a_j, b_j$, and $h(x)$ depend on the
confining potential $V(x)$ and can be found in terms of the unique
solution to an equilibrium problem. Generically $h$ does not vanish
on the intervals $[a_j, b_j]$, and in particular not at the
endpoints $a_j, b_j$, so that the limiting mean eigenvalue density
vanishes like a square root at the endpoints of the spectrum
\cite{KM}.

The two-point eigenvalue correlation kernel in the model (\ref{RMT})
is given by
\begin{equation*}
    K_n(x,y)
        =\frac{e^{-\frac{n}{2}V(x)}
e^{-\frac{n}{2}V(y)}}{x-y}\frac{\kappa_{n-1}}{\kappa_{n}}
(p_n(x)p_{n-1}(y)-p_n(y)p_{n-1}(x)),
\end{equation*}
where the $p_j$'s are polynomials orthonormal with respect to the
weight $e^{-nV}$ on the real line; the leading coefficient of $p_j$
is $\kappa_j$. Scaling limits of the two-point kernel give rise to
well-known limiting kernels such as the sine kernel in the bulk of
the spectrum (where $\psi_V$ is positive) and the Airy kernel near
an endpoint $a_j$ (or $b_j$) where $h(a_j)\neq 0$ (or $h(b_j)\neq
0$) \cite{Deift, DKMVZ2, DKMVZ1}. Near points in the spectrum where
$h(x)=0$, more complicated transcendental kernels appear in double
scaling limits. Near singular edge points where $h(b_j)=0,
h'(b_j=0), h''(b_j)\neq 0$, a kernel related to the $P_I^2$ equation
was obtained \cite{CV2}. Near higher order singular points where
$h(x)\sim c(x-b_j)^{m+\frac{1}{2}}$ for $m$ even (the case where $m$
is odd cannot occur in unitary random matrix ensembles of the form
(\ref{RMT})), it is natural to expect a kernel related to the
$P_I^m$ equation. If one takes double scaling limits where $V$
depends on $n$ in a suitable way, the limiting kernel should depend
on $s$ and also on the $m-1$ time variables $t_1,\ldots, t_{m-1}$.
This is a difference compared to the situation in Theorem
\ref{theorem: KdV}, where only one time variable $t_1$ is non-zero.

\subsubsection*{Outline}
In Section \ref{section: RHP}, we prove the results stated in
Theorem \ref{theorem: PIm} about the equations in the first
Painlev\'e hierarchy. We will construct the pole-free solutions $q$
in terms of a Riemann-Hilbert (RH) problem which depends on $s,t_1,
\ldots, t_{m-1}$, and $m$. We will prove the solvability of this RH
problem for $m$ even and $s,t_1, \ldots, t_{m-1}\in\mathbb R$. This
will imply the absence of real poles for $q$. Using a Lax pair
argument, we will explain the relation between the RH problem, the
equations in the Painlev\'e I hierarchy, and the KdV hierarchy. An
asymptotic analysis of the RH problem will lead to asymptotics for
$q$ as $s\to\pm\infty$.

\medskip

In Section \ref{section 3}, we focus on Theorem \ref{theorem: KdV}.
We will state a well-known RH problem which characterizes solutions
to the Cauchy problem for the KdV equation, and rely on the
techniques developed in \cite{DVZ, DVZ2, CG} to transform this RH
problem to an equivalent one suitable for asymptotic analysis as
$\e\to 0$. The main new point is the construction of a local
parametrix using the model RH problem associated to the pole-free
solution to the $P_{\rm I}^m$ equation studied in the first part.

\section{Pole-free solutions to the Painlev\'e I hierarchy}\label{section:
RHP}

In this section, we will construct the solutions $q$ occurring in
Theorem \ref{theorem: PIm} in terms of a matrix RH problem. At
several points, we will refer to \cite{CV1} where the proof of
Theorem \ref{theorem: PIm} has been given in the case $m=2$, and
where a more detailed exposition can be found.

We consider the following RH problem.

\begin{figure}[t]
\begin{center}
    \setlength{\unitlength}{1truemm}
    \begin{picture}(100,48.5)(0,2.5)
        \put(50,27.5){\thicklines\circle*{.8}}
        \put(49,30){\small $\hat z_0$}

        \put(50,27.5){\line(-2,1){40}}
        \put(50,27.5){\line(-2,-1){40}}
        \put(5,27.5){\line(1,0){80}}

        \put(76,29){$\Gamma_1$}
        \put(17,45){$\Gamma_2$}
        \put(10,29){$\Gamma_3$}
        \put(17,14){$\Gamma_4$}

        \put(55,41){I}
        \put(13,35){II}
        \put(12,19){III}
        \put(55,14){IV}

        \put(30,37.5){\thicklines\vector(2,-1){.0001}}
        \put(30,17.5){\thicklines\vector(2,1){.0001}}
        \put(25,27.5){\thicklines\vector(1,0){.0001}}
        \put(70,27.5){\thicklines\vector(1,0){.0001}}
    \end{picture}
    \caption{The jump contour $\Gamma$.}
    \label{figure: RHP Psi}
\end{center}
\end{figure}
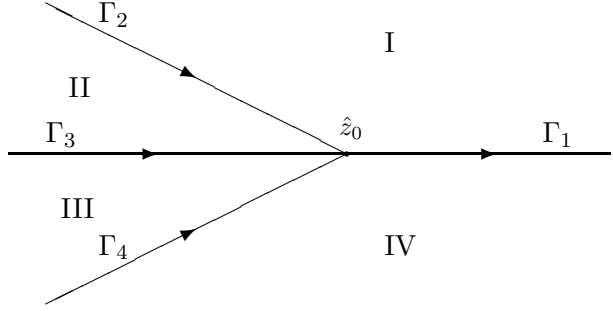

\subsubsection*{RH problem for $\Psi$:}

\begin{itemize}
    \item[(a)] $\Psi:\mathbb{C}\setminus\Gamma\to\mathbb C^{2\times 2}$ is analytic,
    where $\Gamma=\bigcup_{j=1}^4\Gamma_j$ is a contour consisting of
four straight rays
\begin{align*}
    &\Gamma_1:\arg(\zeta-\hat z_0)=0,&\Gamma_2:\arg(\zeta-\hat z_0)=\theta,\\&\Gamma_3:\arg(\zeta-\hat z_0)=\pi,&\Gamma_4:\arg(\zeta-\hat z_0)=-\theta,
\end{align*}
each of them oriented from the left to the right, see Figure
\ref{figure: RHP Psi}. Here $\hat z_0$ can be any real number, and
$\frac{2m+1}{2m+3}\pi<\theta<\pi$.
    \item[(b)] $\Psi$ has continuous boundary values $\Psi_\pm(\zeta)$ when approaching $\zeta\in\Gamma\setminus\{\hat z_0\}$
    from the left or right, and they are related by the jump
    conditions
    \begin{align}
        \Psi_+(\zeta)&=\Psi_-(\zeta)
            \begin{pmatrix}
                1 & 1 \\
                0 & 1
            \end{pmatrix},&& \mbox{for $\zeta\in\Gamma_1$,}\\[1ex]
        \Psi_+(\zeta)&=\Psi_-(\zeta)
            \begin{pmatrix}
                1 & 0 \\
                1 & 1
            \end{pmatrix},&& \mbox{for
            $\zeta\in\Gamma_2\cup\Gamma_4$,}\\[1ex]
        \Psi_+(\zeta)&=\Psi_-(\zeta)
            \begin{pmatrix}
                0 & 1 \\
                -1 & 0
            \end{pmatrix},&& \mbox{for $\zeta\in\Gamma_3$.}
    \end{align}
    \item[(c)] As
    $\zeta\to\infty$, $\Psi$ has an expansion of the form
    \begin{equation}\label{RHP Psi:c}
        \Psi(\zeta)=\zeta^{-\frac{1}{4}\sigma_3}N\left(I+h\sigma_3\zeta^{-1/2}
        +\frac{1}{2}\begin{pmatrix}h^2 & -iq\\iq &
        h^2\end{pmatrix}\zeta^{-1}+\bigO(\zeta^{-3/2})\right)
        e^{-\theta(\zeta;s,t_1, \ldots, t_{m-1})\sigma_3},
    \end{equation}
    where $h,q$ do not depend on $\zeta$, where
    \begin{equation}\label{def N}
        N=\frac{1}{\sqrt 2}
        \begin{pmatrix}
             1 & 1 \\
             -1 & 1
        \end{pmatrix}e^{-\frac{1}{4}\pi i\sigma_3},\qquad \sigma_3=\begin{pmatrix}1&0\\0&-1\end{pmatrix},\end{equation}
        and
        \begin{equation}\label{def theta}
        \theta(\zeta;s,t_1,\ldots, t_{m-1})=\frac{4}{2m+3}\zeta^{\frac{2m+3}{2}}+\sum_{j=1}^{m-1}\frac{4}{2j+1}t_j\zeta^{\frac{2j+1}{2}}\ +s\zeta^{1/2}.
        \end{equation}
\end{itemize}
\begin{remark}\label{remark contour}
The jump contour $\Gamma=\Gamma(\hat z_0,\theta)$ of the RH problem
depends on the values of $\hat z_0$ and $\theta$, but the RH
solution corresponding to $(\hat z_0,\theta)$ can be transformed
directly to the RH solution for any other value of $(\hat z_0',
\theta')$, as long as $\frac{2m+1}{2m+3}\pi<\theta'<\pi$. This can
be done by analytic continuation of the RH solution $\Psi$ across
its jump contour. In this section, we will fix the values of $\hat
z_0=0$ and $\theta=\frac{2m+2}{2m+3}\pi$, but for the asymptotic
analysis in Section \ref{section: asymptotics}, we will need to choose $\hat z_0,\theta$ more carefully.
Near $\hat z_0$, we need to impose that $\Psi$ remains bounded in order to have a unique solution.
\end{remark}

The RH problem depends on $s, t_1, \ldots, t_{m-1}$. For fixed real
values of those parameters, we will show that there exists a
(unique) RH solution $\Psi(\zeta)$ and constants $h$ and $q$ such
that (\ref{RHP Psi:c}) holds. We have $\Psi=\Psi(\zeta;s,t_1,
\ldots, t_{m-1})$, $h=h(s,t_1,\ldots, t_{m-1})$, and
$q=q(s,t_1,\ldots, t_{m-1})$. We will show that $q$ satisfies the
conditions stated in Theorem \ref{theorem: PIm}.

\begin{remark}
If $m=0$ and $s=0$, this is a well-known RH problem which can be
solved explicitly in terms of the Airy function and its derivative.
For $m=2$, $\tilde\zeta=60^{2/7}\zeta$, $T=-4\cdot 60^{-3/7}t_1$,
and $X=60^{-1/7}s$, we obtain the RH problem studied in \cite{CV1}
related to the pole-free solution to the fourth order equation
(\ref{PI2}) which is equivalent to equation (\ref{PIm}) for $m=2$.
\end{remark}

\subsection{Solvability of the RH problem}
In order to prove the solvability of the RH problem for $\Psi$, we
perform the transformation
\begin{equation}\label{def Phi}
\Phi(\zeta)=\begin{pmatrix}1&0\\h&1\end{pmatrix}\Psi(\zeta)e^{\theta(\zeta)\sigma_3},
\end{equation}
where we have suppressed the dependence on $s,t_1, \ldots, t_{m-1}$
in our notations. This leads to a slightly modified RH problem for
$\Phi$.
\subsubsection*{RH problem for $\Phi$:}
\begin{itemize}
    \item[(a)] $\Phi$ is analytic in $\mathbb{C}\setminus\Gamma$.
    \item[(b)] We have the jump relations
    \begin{align}
        \Phi_+(\zeta)&=\Phi_-(\zeta)
            \begin{pmatrix}
                1 & e^{-2\theta(\zeta)} \\
                0 & 1
            \end{pmatrix},&& \mbox{for $\zeta\in\Gamma_1$,}\\[1ex]
        \Phi_+(\zeta)&=\Phi_-(\zeta)
            \begin{pmatrix}
                1 & 0 \\
                e^{2\theta(\zeta)} & 1
            \end{pmatrix},&& \mbox{for
            $\zeta\in\Gamma_2\cup\Gamma_4$,}\\[1ex]
        \Phi_+(\zeta)&=\Phi_-(\zeta)
            \begin{pmatrix}
                0 & 1 \\
                -1 & 0
            \end{pmatrix},&& \mbox{for $\zeta\in\Gamma_3$.}
    \end{align}
    \item[(c)] As
    $\zeta\to\infty$, $\Phi$ has an expansion of the form
    \begin{equation}\label{RHP Phi:c}
        \Phi(\zeta)=\left(I+\frac{A_1}{\zeta}+\bigO(\zeta^{-2})\right)\zeta^{-\frac{1}{4}\sigma_3}N.
    \end{equation}
    It is straightforward to verify by (\ref{RHP Psi:c}) and (\ref{def Phi}) that $q$ is given by
\begin{equation}\label{qA1}
q=A_{1,12}^2-2A_{1,11}.
\end{equation}
\end{itemize}

Following the general procedure developed in \cite{FokasZhou,
FokasMuganZhou, DKMVZ2}, the solvability of a large class of RH
problems is equivalent to the fact that a homogeneous version of the
RH problem has no solution except the trivial zero solution. This
follows from the observation that a singular integral operator
associated to the RH problem is a Fredholm operator of index zero.
This procedure is explained in detail in \cite{DKMVZ2, IKO} and in
\cite{CV1} for the above RH problem in the case where $m=2$. All of
the arguments work for general $m$ and general real values of the
$t_j$'s and we will not repeat them here. In order to have
solvability of the RH problem for $\Phi$, it is sufficient to prove
the following vanishing lemma.
\begin{lemma}{\bf (Vanishing lemma)}
Let $s,t_1, \ldots, t_{m-1}\in\mathbb R$, let $\Phi_0$ satisfy
conditions (a) and (b) of the RH problem for $\Phi$, and let in
addition
\begin{equation}
\Phi_0(\zeta)=\bigO(\zeta^{-3/4}),\qquad\mbox{ as $\zeta\to
\infty$.}\end{equation} Then $\Phi_0\equiv 0$.
\end{lemma}
\begin{proof}
The proof of the vanishing lemma is almost the same as in \cite{CV1,
DKMVZ2} and goes as follows. We first collapse the jump contour
$\Gamma$ in the RH problem for $\Phi$ to the real line. We do this
by defining
\begin{align*}
    &A(\zeta)=
        \Phi_0(\zeta)
        \begin{pmatrix}
            0 & -1 \\
            1 & 0
        \end{pmatrix},& \mbox{for $0<\arg\zeta<\theta$,}
        \\[4ex]
        &A(\zeta)=
        \Phi_0(\zeta)
        \begin{pmatrix}
            1 & 0 \\
            e^{2\theta(\zeta)} & 1
        \end{pmatrix}
        \begin{pmatrix}
            0 & -1 \\
            1 & 0
        \end{pmatrix},& \mbox{for
        $\theta<\arg\zeta<\pi$,}\\[4ex]
        &A(\zeta)=
        \Phi_0(\zeta)
        \begin{pmatrix}
            1 & 0 \\
            -e^{2\theta(\zeta)} & 1
        \end{pmatrix},& \mbox{for
        $-\pi<\arg\zeta<-\theta$,}\\[4ex]
        &A(\zeta)=
        \Phi_0(\zeta), & \mbox{for $-\theta<\arg\zeta<0$.}
\end{align*}
Then one verifies using the jump relations for $\Phi_0$ that $A$ is
analytic in $\mathbb C\setminus\mathbb R$, and that it has the jump
properties
    \begin{align}
        \label{RHP A: b1}
        A_+(\zeta)&=A_-(\zeta)\begin{pmatrix}
            1 & -e^{2\theta_+(\zeta)} \\
            e^{2\theta_-(\zeta)} & 0
        \end{pmatrix},&\mbox{for $\zeta\in\mathbb R_-$,}\\[1ex]
        \label{RHP A: b2}
        A_+(\zeta)&=A_-(\zeta)\begin{pmatrix}
            e^{-2\theta(\zeta)} & -1 \\
            1 & 0
        \end{pmatrix},&\mbox{for $\zeta\in\mathbb R_+$.}
    \end{align}
    As $\zeta\to\infty$, the behavior of $A$ follows from the asymptotics for $\Phi_0$ and the fact that the
    exponentials in
    the definition of $A$ are uniformly bounded. We have
    \begin{equation}\label{RHP A:c}
    A(\zeta)=\bigO(\zeta^{-3/4}),\qquad \mbox{ as
    $\zeta\to\infty$},\end{equation}
    uniformly for $\zeta\in\mathbb C\setminus\mathbb
    R$.
Then we define $Q(\zeta)=A(\zeta)A^H(\bar\zeta)$, where $A^H$ is the
Hermitian conjugate of the matrix $A$. $Q$ is clearly analytic in
the upper half plane, because $A$ is analytic in the upper and lower
half plane, and it is  continuous up to $\mathbb{R}$. As
$\zeta\to\infty$, we have $Q(\zeta)=\bigO(\zeta^{-3/2})$, which
implies that
\begin{equation}\int_{\mathbb{R}}Q_+(\xi)d\xi=0\end{equation}
by Cauchy's theorem, and using (\ref{RHP A: b1})-(\ref{RHP A: b2})
we obtain
\begin{multline}\label{eqQ}
    \int_{\mathbb{R}^-}A_-(\xi)
    \begin{pmatrix}
        1 & -e^{2\theta_+(\xi;s,t_1, \ldots, t_{m-1})} \\
        e^{2\theta_-(\xi;s,t_1, \ldots, t_{m-1})} & 0
    \end{pmatrix}A^H_-(\xi)d\xi\\+\int_{\mathbb{R}^+}A_-(\xi)
        \begin{pmatrix}
            e^{-2\theta(\xi;s,t_1, \ldots, t_{m-1})} & -1 \\
            1 & 0
        \end{pmatrix}
        A^H_-(\xi)d\xi=0.
\end{multline}
For real values of $s,t_1, \ldots, t_{m-1}$, we have
\[\overline{\theta_+(\xi;s,t_1, \ldots, t_{m-1})}=\theta_-(\xi;s,t_1,
\ldots, t_{m-1}),\qquad \mbox{ for $\xi<0$,}\] and adding
(\ref{eqQ}) to its Hermitian conjugate, we obtain
\begin{equation}
    \int_{\mathbb{R}^-}A_-(\xi)
    \begin{pmatrix}
        2 & 0 \\
        0 & 0
    \end{pmatrix}A^H_-(\xi)d\xi+\int_{\mathbb{R}^+}A_-(\xi)
        \begin{pmatrix}
            2e^{-2\theta(\xi;s,t_1, \ldots, t_{m-1})} & 0 \\
            0 & 0
        \end{pmatrix}
        A^H_-(\xi)d\xi=0.
\end{equation}
Since $e^{-2\theta(\xi;s,t_1, \ldots, t_{m-1})}>0$, this implies
that the first column of $A_-$ is identically zero (because it is
continuous). The jump conditions (\ref{RHP A: b1})-(\ref{RHP A: b2})
can be used to prove that the second column of $A_+$ vanishes as
well.

Writing out the jump relations for the entries of $A$ for which we
have not yet proved that they vanish, the RH problem for $A$
decouples into two scalar RH problems. The same argument as in
\cite[Step 3 of Section 5.3]{DKMVZ2} shows that those scalar RH
problems have only the zero solution, and that $A\equiv 0$.
Consequently we have that $\Phi_0\equiv 0$, which proves the
vanishing lemma.
\end{proof}

As a consequence of the vanishing lemma, the RH problem for $\Phi$,
and thus also for $\Psi$ (one can invert the transformation defined
by (\ref{def Phi})), has a solution for real values of the
parameters $s,t_1,\ldots, t_{m-1}$.

It follows from the general theory of RH problems that the subset of
$\mathbb C^m$ of values $(s,t_1, \ldots, t_{m-1})$ where the RH
problem is solvable, is an open set, and this implies in particular
that, for any $(s,t_1, \ldots, t_{m-1})\in\mathbb R^m$, there exists
a neighborhood in $\mathbb C^m$ such that the RH problem is solvable
also in this neighborhood. Moreover if this neighborhood is chosen
sufficiently small, condition (\ref{RHP Psi:c}) is valid uniformly,
see \cite{CV1}, and $\Psi$, $q$, and $h$ depend analytically on each
of the variables $s,t_1, \ldots, t_{m-1}$. Values of $(s,t_1,\ldots,
t_{m-1})$ where the RH problem is not solvable, correspond to poles
of $q$.

\subsection{Relation between the RH problem, the Painlev\'e I
hierarchy, and the KdV hierarchy}

We will show that the function $q$ appearing in the asymptotic
expansion (\ref{RHP Psi:c}) solves equation $P_{\rm I}^m$, and that
it also solves the equations in the KdV hierarchy. This follows from
Lax pair arguments which are rather standard, see e.g.\ \cite{FN,
FIKN, IN} in general and \cite{KudrS, Takasaki} for the Lax pair
corresponding to the first Painlev\'e hierarchy. We recall the arguments briefly for the
reader's convenience.

\subsubsection{Lax pair in $s$, $\zeta$, and $t_k$} Since $\Psi$ is
differentiable in $s$, we can define
\begin{equation}
\label{def L A}L(\zeta)=\Psi_s(\zeta)\Psi(\zeta)^{-1},\qquad
A(\zeta)=\Psi_\zeta(\zeta)\Psi(\zeta)^{-1}.
\end{equation}
The jump matrices for $\Psi$ do not depend on $\zeta$ and $s$,
so $L$ and $A$ are entire functions in the complex plane, and because
of the asymptotics (\ref{RHP Psi:c}), they are polynomials in
$\zeta$. For $L$, one deduces directly from (\ref{RHP Psi:c}) that
\begin{equation}
L(\zeta)=\begin{pmatrix}0&1\\
\zeta+(h_s+q)&0\end{pmatrix}=\begin{pmatrix}0&1\\
\zeta+2q&0\end{pmatrix},
\end{equation}
where it was used in the latter equation that $h_s=q$. This follows
by developing $\Psi_s\Psi^{-1}$ as $\zeta\to\infty$ and imposing
that the $12$-entry of the $\zeta^{-1}$-term is zero. For the
$\zeta$-derivative, we use (\ref{RHP Psi:c}) to conclude that $A$ is
a polynomial of degree $m+1$,
\begin{equation}\label{A polynomial}
A(\zeta)=\sum_{j=0}^{m+1}A_j \zeta^j,
\end{equation}
where the $A_j$'s depend on $s, t_1, \ldots, t_{m-1}$ but not on
$\zeta$. Let us take a closer look at the $12$-entry
$\beta(\zeta):=A_{12}(\zeta)$: by (\ref{RHP Psi:c}) and (\ref{def L
A}), it has the form
\begin{equation}\label{beta poly}
\beta(\zeta)=\beta^{(m+1)}(\zeta)+\sum_{k=1}^{m-1}t_k\beta^{(k)}(\zeta),
\end{equation}
where $\beta^{(k)}$ is a polynomial of degree $k-1$ which can be written in the form
\begin{equation}
\beta^{(k)}(\zeta)=\sum_{j=0}^{k-1}\frac{\mathcal L_{k-j-2}}{2}\zeta^j,
\end{equation}
where $\mathcal L_j$ is independent of the value of $k$ in the above
formula, and $\mathcal L_{-1}=4, \mathcal L_0=-4q$. Similar formulas
can be obtained for the other entries but are not needed.

Now, as in \cite{Takasaki}, write
\begin{equation}
A(\zeta)=\begin{pmatrix}\alpha&\beta\\ \gamma &-\alpha\end{pmatrix}
\end{equation}
(it is easily verified that $\Tr A\equiv 0$). Then $\Psi_{\zeta
s}=\Psi_{s\zeta}$ implies the compatibility condition
\begin{equation}
L_\zeta-A_s+[L,A]=0,
\end{equation}
and if we write this down entry-wise, we find
\begin{align}
&\label{alpha}\alpha_s+\beta(\zeta+2q)-\gamma=0,\\
&\beta_s+2\alpha=0,\\
&\gamma_s-1-2\alpha(\zeta+2q)=0.
\end{align}
We solve the second equation for $\alpha$ and
then the first for $\gamma$, and we substitute their values in the
third equation. This gives
\begin{equation}\label{eq beta}
\frac{1}{2}\beta_{sss}-2\beta_s(\zeta+2q)-2q_s\beta+1=0.
\end{equation}
The degree $m+1$ term in (\ref{eq beta}) is trivial, and the degree
$m$ term reads $4q_s+\mathcal L_{0,s}=0$, which we knew already
since $\mathcal L_0=-4q$. From the term of degree $j$ with $1\leq
j\leq m-1$, we get
\begin{equation}\label{eq betaj}
\frac{1}{2}\mathcal L_{m-j-1,sss}-2\mathcal L_{m-j,s}-4q\mathcal
L_{m-j-1,s}-2q_s\mathcal L_{m-j-1}=0,
\end{equation}
and the constant term in (\ref{eq beta}) gives
\begin{equation}\label{eq beta0}
1+\frac{1}{4}\mathcal L_{m-1,sss}-2q\mathcal L_{m-1}'-q_s\mathcal
L_{m-1}+\sum_{j=1}^{m-1}t_j\left(\frac{1}{4}\mathcal
L_{j-2,sss}-2q\mathcal L_{j-2}'-q_s\mathcal L_{j-2}\right)=0,
\end{equation}
or
\begin{equation}\label{diffPIm}1+\frac{d}{ds}\mathcal L_m+\sum_{j=1}^{m-1}t_j\frac{d}{ds}\mathcal L_{j-1}=0.\end{equation} Integrating this equation gives the
$P_I^m$ equation (\ref{PIm}). One shows using the asymptotic
behavior of solutions to the Schr\"odinger equation $\Psi_s=L\Psi$
that the constants of integration when integrating (\ref{diffPIm})
and (\ref{eq betaj}) are zero. This proves that $q$ solves the
$P_I^m$ equation (\ref{PIm}).

\medskip

The jump matrices for $\Psi$ are also independent of $t_1,\ldots,
t_{m-1}$, and consequently $B^{(k)}=\Psi_{t_k}\Psi^{-1}$ is a
polynomial of degree $k+1$ in $\zeta$. Exploiting the compatibility
of the $t_k$-derivative with the $s$-derivative, an analogous
argument as before leads to the time flow
\begin{equation}\label{KdVk}
q_{t_k}+\frac{1}{2k+1}\frac{d}{ds}\mathcal L_k=0.
\end{equation}

\subsection{Asymptotics for $q$}\label{section: asymptotics}

We will now analyze the RH problem for $\Psi$ asymptotically as
$s\to\pm\infty$. We will use the
Deift/Zhou steepest descent method to obtain asymptotics for $\Psi$.
For $m=2$, this analysis has been done in \cite{CV1}, see also
\cite{Kapaev, Kapaev2}. For $m>2$, the general approach remains the same, but
in particular the construction of the $g$-function and the contour
deformation are more delicate.

\subsubsection{Re-scaling of $\Psi$} Until now, we have always
considered the RH problem for $\Psi$ corresponding to $\hat z_0=0$
and $\theta=\frac{2m+2}{2m+3}\pi$, see Remark \ref{remark contour}.
For the asymptotic analysis of the RH problem, we will need other
values of $\hat z_0$ and $\theta$ which we will specify later.

Define
\begin{equation}\label{def Y}
Y(\zeta)=\begin{pmatrix}1&0\\h&0\end{pmatrix}\Psi(|s|^{\frac{1}{m+1}}\zeta;s,t_1,
\ldots, t_{m-1}).
\end{equation}

\subsubsection*{RH problem for $Y$:}

\begin{itemize}
    \item[(a)] $Y$ is analytic in $\mathbb{C}\setminus\Gamma(z_0, \theta)$,
    where $z_0=|s|^{-\frac{1}{m+1}}\hat z_0$.
    \item[(b)] We have the jump conditions
    \begin{align}
        Y_+(\zeta)&=Y_-(\zeta)
            \begin{pmatrix}
                1 & 1 \\
                0 & 1
            \end{pmatrix},&& \mbox{for $\zeta\in\Gamma_1$,}\\[1ex]
        Y_+(\zeta)&=Y_-(\zeta)
            \begin{pmatrix}
                1 & 0 \\
                1 & 1
            \end{pmatrix},&& \mbox{for
            $\zeta\in\Gamma_2\cup\Gamma_4$,}\\[1ex]
        Y_+(\zeta)&=Y_-(\zeta)
            \begin{pmatrix}
                0 & 1 \\
                -1 & 0
            \end{pmatrix},&& \mbox{for $\zeta\in\Gamma_3$.}
    \end{align}
    \item[(c)] As
    $\zeta\to\infty$, we have
    \begin{equation}\label{RHP Y:c}
        Y(\zeta)=\left(I+|s|^{-\frac{1}{m+1}}A_1\zeta^{-1}+\bigO(\zeta^{-2})\right)|s|^{-\frac{1}{4m+4}\sigma_3}\zeta^{-\frac{\sigma_3}{4}}N
        e^{-|s|^{\frac{2m+3}{2m+2}}\hat\theta(\zeta;s,t_1, \ldots,
        t_{m-1})\sigma_3}.
    \end{equation}
    where
        \begin{align}\label{def hattheta}
        &\hat\theta(\zeta;s,t_1,\ldots, t_{m-1})=|s|^{-\frac{2m+3}{2m+2}}\theta(|s|^{\frac{1}{m+1}}\zeta;s,t_1,\ldots,
        t_{m-1})\\
        &\label{hat theta 0}\qquad=\frac{4}{2m+3}\zeta^{\frac{2m+3}{2}}+\sgn(s)\zeta^{1/2}+\sum_{j=1}^{m-1}\frac{4}{2j+1}t_j|s|^{\frac{j-m-1}{m+1}}\zeta^\frac{2j+1}{2}.
        \end{align}
\end{itemize}
The matrix $A_1$ in (\ref{RHP Y:c}) is the same as in (\ref{RHP
Phi:c}), so (\ref{qA1}) holds.

\subsubsection{Construction of the $g$-function}
We proceed with our analysis in the case where $t_1, \ldots, t_{m-1}=0$. Straightforward modifications described in Remark \ref{remark tj} allow us to treat also the general case where the $t_j$'s are fixed. We search a
$g$-function $g=g(\zeta;s)$ of the form
\begin{equation}\label{def g}
g(\zeta)=(\zeta-z_0)^{3/2}p\left(\frac{\zeta}{z_0}\right),
\end{equation}
where $p(z)$ is a polynomial of degree $m$. Its coefficients and $z_0$ are uniquely determined by the
condition
\begin{equation}\label{as g}
g(\zeta)=\hat\theta(\zeta)+\bigO(\zeta^{-1/2}),\qquad \mbox{ as
$\zeta\to\infty$.}
\end{equation} We have
\begin{equation}\label{hat theta 01}
\hat\theta(\zeta)=\frac{4}{2m+3}\zeta^{\frac{2m+3}{2}}+\sgn(s)\zeta^{1/2},
\end{equation}
and because
\begin{equation}
(\zeta-z_0)^{-3/2}=\zeta^{-3/2}\left(1+\sum_{j=1}^{\infty}\frac{(2j+1)!!}{2^j.j!}z_0^j\zeta^{-j}\right),\qquad
\mbox{ as $\zeta\to\infty$},
\end{equation}
(\ref{as g}) requires us to take $p$ of the form
\begin{equation}
p(z)=\frac{4}{2m+3}z_0^{m}\sum_{j=0}^{m}c_jz^{m-j},\qquad
c_j=\frac{(2j+1)!!}{2^j.j!}.
\end{equation}
The missing condition in order to have (\ref{as g}) is
\begin{equation}\label{z0}
z_0=-\sgn(s)\left(\frac{2^{m-1}(m+1)!}{(2m+1)!!}\right)^{\frac{1}{m+1}}.
\end{equation}
With this choice of $g$, we have
\begin{equation}\label{asymptotics exp theta g}
        e^{|s|^{\frac{2m+3}{2m+2}}(g(\zeta)-\hat\theta(\zeta))\sigma_3}=
        I+\sum_{k=1}^\infty d_k\sigma_3^k\zeta^{-k/2},\qquad\mbox{ as $\zeta\to\infty$,}
\end{equation}
where the coefficients $d_k$ can be calculated but are not important. Since the
determinant of the left hand side
 of (\ref{asymptotics exp theta g}) is $1$, we have
\begin{equation}\label{definition: d2}
    d_2=\frac{1}{2}d_1^2.
\end{equation}

\begin{remark}\label{remark tj}
If the $t_j$'s do not vanish, we need to modify $p$ and $z_0$, but
 (\ref{as g}) still determines the coefficients $c_j$
of the polynomial
\begin{equation}
p(z)=\frac{4}{2m+3}z_0^{m}\sum_{j=0}^{m}c_j z^{m-j}
\end{equation}
uniquely. We don't need their explicit form, it is sufficient to
know that
\begin{align}&\label{cjxt}
c_j=\frac{(2j+1)!!}{2^j.j!}+\bigO(|s|^{-\frac{2}{m+1}}),&
\mbox{ as $|s|\to\infty$},\\
&\label{z0t}z_0=-\sgn(s)\left(\frac{2^{m-1}(m+1)!}{(2m+1)!!}\right)^{\frac{1}{m+1}}+\bigO(|s|^{-\frac{2}{m+1}}),&\mbox{
as $|s|\to\infty$.}
\end{align}
To prove Theorem \ref{theorem: PIm}, we only deal with the case
where $t_1=\ldots=t_{m-1}=0$, but the entire analysis done below can
be easily generalized as long as the $t_j$'s remain bounded. Formula
(\ref{as q2}) will then follow from (\ref{as q3}) together with
(\ref{z0t}).
\end{remark}

\begin{proposition}\label{prop g}We have
\begin{align}
&\label{ineq1}g(\zeta)>0,&\mbox{ for $\zeta>z_0$,}\\
&\label{ineq2}\Im g_+'(\zeta)>0,&\mbox{ for $\zeta<z_0$.}
\end{align}
\end{proposition}
\begin{proof}
For the first equality, we first prove that $p(z)>0$ for
$z\in(-\infty,-3/2]\cup[-1,+\infty)$. Since the coefficients of $p$
are positive, this is clear for $z$ positive. For $-1<z<0$, we have
\[
p(z)=\frac{4}{2m+3}z_0^{m}\left((c_m-c_{m-1}|z|)+z^2(c_{m-2}-c_{m-3}|z|)+\ldots
+ z^{m-2}(c_{2}-c_1|z|) + c_0z^m\right),
\]
and all the terms in this expression are positive since the $c_j$'s
increase with $j$. For $z<-3/2$, a similar argument shows that
$p(z)>0$ (the terms in the sum are alternating and their absolute
value increases with the degree). This implies that $g(\zeta)>0$ for
$\frac{\zeta}{z_0}\in (-\infty,-3/2]\cup[-1,+\infty]$. In the case
where $z_0>0$ (or $s<0$), this implies (\ref{ineq1}), for $z_0<0$,
we still need an estimate for $\frac{\zeta}{z_0}\in (-3/2,-1)$, or
$\zeta\in(-z_0,-\frac{3}{2}z_0)$. We already know that $g$ is
positive at both endpoints of the interval, so it suffices to prove
that $g(\zeta)$ is monotonic on $(-z_0,-\frac{3}{2}z_0)$.

\medskip

Write $g'(\zeta)=(\zeta-z_0)^{1/2}q(\frac{\zeta}{z_0})$ with
\[q(z)=2z_0^{m}\left(z^{m}+\sum_{j=1}^{m}b_j
z^{m-j}\right),\qquad b_j=\frac{(2j-1)!!}{2^j j!}.\] We will show
that $q(z)$ is strictly positive for all $z\in\mathbb R$, which
implies (\ref{ineq2}), and also the fact that $g(\zeta)$ is
monotonic on $(-z_0,-\frac{3}{2}z_0)$, which completes the proof of
the first inequality.

\medskip

For $z\in(-\infty,-1]\cup[-\frac{1}{2},+\infty)$, it is not
difficult to see that $q(z)$ is positive. Indeed, the terms in the sum are either positive or alternating and monotonic with increasing degree. For
$-1<z<-\frac{1}{2}$, write
\begin{multline}
\frac{z_0^{-m}}{2}z^{2-m}q(z)=\left(z^2+\frac{1}{2}z+\frac{3}{8}\right)+\sum_{k=2}^{\frac{m}{2}}
\left(b_{2k-1}z^{-2k+3}+b_{2k}z^{-2k+2}\right).
\end{multline}
The first part is bigger than $\frac{5}{16}$, and each part of the
sum reaches its minimal value on $[-1,-1/2]$ at $-1$. Since the
$b_j$'s are decreasing, rearranging the alternating terms, we obtain
the estimate
\[
\frac{z_0^{-m}}{2}z^{2-m}q(z)\geq\frac{5}{16}+\sum_{k=2}^{\frac{m}{2}}
\left(-b_{2k-1}+b_{2k}\right)\geq \frac{5}{16}-\frac{5}{16}+b_m>0.
\]
\end{proof}

The above result enables us to apply the Cauchy-Riemann conditions,
which leads to the following corollary.
\begin{corollary}\label{corollary}
There exists a $\theta>0$ such that $\Re g(\zeta)<0$ if
$\arg(\zeta-z_0)=\pi\pm\theta$.
\end{corollary}
We still have the freedom to choose the values of $\theta$ and $z_0$ that determine the jump contour $\Gamma(z_0,\theta)$ for $Y$. We take $z_0$ as in (\ref{z0}) and $\theta$ such that the above corollary holds.

\subsubsection{Normalization of the RH problem} Define
\begin{equation}\label{def S}S(\zeta)=\begin{pmatrix}1&0\\d_1|s|^{\frac{1}{2m+2}}&1\end{pmatrix}Y(\zeta)e^{|s|^{\frac{2m+3}{2m+2}}g(\zeta)\sigma_3},\end{equation}
so that we have
\subsubsection*{RH problem for $S$:}
\begin{itemize}
    \item[(a)] $S$ is analytic in $\mathbb{C}\setminus\Gamma(z_0,\theta)$.
    \item[(b)]  For $\zeta\in\Gamma(z_0,\theta)$,
    \begin{align}\label{RHP S:b1}
        S_+(\zeta)&=S_-(\zeta)
            \begin{pmatrix}
                1 & e^{-2|s|^{\frac{2m+3}{2m+2}}g(\zeta)} \\
                0 & 1
            \end{pmatrix},&& \mbox{for $\zeta\in\Gamma_1$,}\\[1ex]
        S_+(\zeta)&=S_-(\zeta)
            \begin{pmatrix}
                1 & 0 \\
                e^{2|s|^{\frac{2m+3}{2m+2}}g(\zeta)} & 1
            \end{pmatrix},&& \mbox{for
            $\zeta\in\Gamma_2\cup\Gamma_4$,}\\[1ex]
        S_+(\zeta)&=S_-(\zeta)\label{RHP S:b3}
            \begin{pmatrix}
                0 & 1 \\
                -1 & 0
            \end{pmatrix},&& \mbox{for $\zeta\in \Gamma_3$.}
    \end{align}
    \item[(c)] As $\zeta\to\infty$, $S$ behaves like
    \begin{equation}\label{RHP S: c}
        S(\zeta)=
        \left[I+B_1\zeta^{-1}+\bigO(\zeta^{-2})\right]
            |s|^{-\frac{\sigma_3}{4m+4}}\zeta^{-\frac{\sigma_3}{4}} N,
    \end{equation}
    and $q$ is given by
\begin{equation}\label{qB1} q=|s|^{\frac{2}{m+1}}B_{1,12}^2-2|s|^{\frac{1}{m+1}}B_{1,11}.
\end{equation}

\end{itemize}
As $s\to\pm\infty$, the jump matrices for $S$ tend to constant
matrices except near $z_0$: indeed on $\Gamma_1,\Gamma_2,$ and
$\Gamma_4$ we have exponentially fast convergence to the identity
matrix by Proposition \ref{prop g} and Corollary \ref{corollary},
and on $\Gamma_3$, the jump matrix is identically equal to
$\begin{pmatrix}0&1\\-1&0\end{pmatrix}$.

\subsubsection{Global parametrix}

In the limit $s\to\pm\infty$, if we ignore a small neighborhood of
$z_0$, the RH problem for $S$ reduces to a RH problem with a jump
only on $(-\infty,z_0)$. We can explicitly construct a solution
$P^{(\infty)}$ to this RH problem.

\subsubsection*{RH problem for $P^{(\infty)}$:}
\begin{itemize}
    \item[(a)] $P^{(\infty)}$ is analytic in $\mathbb{C}\setminus (-\infty,z_0]$.
    \item[(b)]  We have \begin{equation}
        P^{(\infty)}_+(\zeta)=P^{(\infty)}_-(\zeta)
                        \begin{pmatrix}
                0 & 1 \\
                -1 & 0
            \end{pmatrix},\qquad \mbox{for $\zeta\in (-\infty,z_0)$.}
    \end{equation}
    \item[(c)] As $\zeta\to\infty$,
    \begin{equation}\label{RHP Pinfty: c1}
        P^{(\infty)}(\zeta)=\left(I+\bigO(\zeta^{-1})\right)|s|^{-\frac{1}{4m+4}\sigma_3}\zeta^{-\frac{\sigma_3}{4}}N.
    \end{equation}
\end{itemize}
This RH problem can easily be solved explicitly: if we take
\begin{equation}\label{def Pinfty}
    P^{(\infty)}(\zeta)=
    |s|^{-\frac{\sigma_3}{4m+4}}(\zeta- z_0)^{-\frac{\sigma_3}{4}}N,
\end{equation}
one verifies that $P^{(\infty)}$ satisfies the required conditions.
The asymptotic condition (\ref{RHP Pinfty: c1}) can be specified as
\begin{equation}\label{RHP Pinfty: c}
        P^{(\infty)}(\zeta)=\left(I+\frac{z_0}{4\zeta}\sigma_3+\bigO(\zeta^{-2})\right)|s|^{-\frac{\sigma_3}{4m+4}}\zeta^{-\frac{\sigma_3}{4}}N.
    \end{equation}
 We will show that $P^{(\infty)}$ determines the leading order
asymptotics of $S$ and thus indirectly of the $P_I^m$ solution $q$,
by (\ref{qB1}). Therefore we first need to know that there exists a
local parametrix near $z_0$ which matches with the global
parametrix.

\subsubsection{Local parametrix near $z_0$} Let us fix a small neighborhood $U$ of $z_0$, for example a
small disk. Given $m$ and $\sgn(s)$, we take $U$ fixed for $|s|$
sufficiently large. Near $z_0$, the $g$-function vanishes like
$c(\zeta-z_0)^{3/2}$. Following a well understood procedure, one can
explicitly construct a local parametrix in $U$ in terms of the Airy
function and its derivative. We refer to \cite{CV1} for the explicit
construction (and to \cite{Deift, DKMVZ2, DKMVZ1} for similar
constructions). The only thing that we need here, is the existence
of a local parametrix which satisfies the RH problem

\subsubsection*{RH problem for $P$:}
\begin{itemize}
\item[(a)]$P$ is analytic in $\overline U\setminus \Gamma(z_0,\theta)$,
\item[(b)]for $\zeta\in \Gamma(z_0,\theta)\cap U$, $P$ satisfies exactly the same jump conditions than $S$ (see (\ref{RHP S:b1})-(\ref{RHP S:b3})),
\item[(c)]for $\zeta\in\partial U$, we have
\begin{equation}\label{matching}
P(\zeta)P^{(\infty)}(\zeta)^{-1}=I+\bigO(|s|^{-1}),\qquad\mbox{ as
$s\to\pm\infty$}.
\end{equation}
\end{itemize}

\subsubsection{Final transformation} Define
\begin{equation}\label{def R}
R(\zeta)=\begin{cases} S(\zeta)P(\zeta)^{-1},&\mbox{ for
$\zeta\in U$,}\\
S(\zeta)P^{(\infty)}(\zeta)^{-1},&\mbox{ for $\zeta\in\mathbb
C\setminus \overline U$}.
\end{cases}
\end{equation}
Then $R$ is analytic in the interior of $U$ and across
$(-\infty,z_0)$ because the jumps of $S$ cancel against the jumps of
the parametrices $P$ and $P^{(\infty)}$. On the boundary of $U$, the
matching of the local parametrix with the global parametrix, see
(\ref{matching}), implies that $R$ has a jump that is
$I+\bigO(|s|^{-1})$ as $s\to\pm\infty$.

\subsubsection*{RH problem for $R$:}
\begin{itemize}
\item[(a)] $R$ is analytic in $\mathbb C\setminus \Sigma_R$, with $\Sigma_R=(\Gamma\setminus U)\cup\partial
U$,
\item[(b)] $R_+(\zeta)=R_-(\zeta)v_R(\zeta)$ for $\zeta\in\Sigma_R$, where
\begin{align}
&v_R(\zeta)=I+\bigO(|s|^{-1}),&&\mbox{ as $s\to\pm\infty$, for $\zeta\in\partial U$},\\
&v_R(\zeta)=I+\bigO(e^{-c|s|(|\zeta|+1)}),&&\mbox{ as
$s\to\pm\infty$, for $\zeta\in\Sigma_R\setminus \partial U$},
\end{align}
\item[(c)] There exists a matrix $R_1=R_1(s)$ such that
$R(\zeta)=I+R_1\zeta^{-1}+\bigO(\zeta^{-2})$ as $\zeta\to\infty$.
\end{itemize}
It is a standard fact that the solution to a RH problem of this form
(with small jump matrices and normalized at infinity) is close to
the identity matrix \cite{DKMVZ1}: we have
\begin{equation}
R(\zeta)=I+\bigO(|s|^{-1}),\qquad\mbox{ as $s\to\pm\infty$,}
\end{equation}
and for the residue matrix at infinity we have
\begin{equation}
R_1(s)=\bigO(|s|^{-1}),\qquad \mbox{ as $s\to\pm\infty$}.
\end{equation}
Using (\ref{RHP Pinfty: c}) and (\ref{def R}), one derives the
identity
\begin{equation}
R_1=B_1-\frac{z_0}{4}\sigma_3,
\end{equation}
and this implies that
\begin{equation}
        B_{1,11} = \frac{z_0}{4}+\bigO(|s|^{-1}), \qquad
 B_{1,12} = \bigO(|s|^{-1}),\qquad \mbox{ as $s\to\pm\infty$}.
    \end{equation}
By (\ref{qB1}) we obtain
\begin{equation}\label{as q3}
q(s)=-\frac{z_0}{2}|s|^{\frac{1}{m+1}}+\bigO(|s|^{-\frac{m}{m+1}}),\qquad
\mbox{ as $s\to\pm\infty$},
\end{equation}
Now (\ref{as q}) follows from (\ref{z0t}), and the more general
asymptotic formula (\ref{as q2}) follows from (\ref{z0t}).

\section{Critical behavior of solutions to the KdV equation}\label{section
3} In this section we will prove Theorem \ref{theorem: KdV} and show
that the pole-free solutions to the even members of the Painlev\'e I
hierarchy describe the critical behavior of solutions to the KdV
equation. A RH procedure to obtain asymptotics for KdV solutions was
developed in \cite{DVZ, DVZ2}. In \cite{CG} the method was used to
prove Theorem \ref{theorem: KdV} in the generic case where $m=2$.
For the sake of brevity and because many of the arguments are valid
also for $m>2$, we will refer to this paper at several points. In
the RH analysis of the KdV RH problem, we will construct auxiliary
matrix functions $S$, $P^{(\infty)}$, $P$, and $R$. They are not the
same functions as in the previous section, we hope this does not
cause any confusion.

\subsection{RH problem for the KdV equation}
Given initial data $u_0(x)$ satisfying the conditions specified in
the introduction (i.e.\ $u_0(x)$ is negative, real analytic, has a
single negative hump, and decays sufficiently fast at $\pm\infty$),
we are interested in the solution $u(x,t,\e)$ to the Cauchy problem
for the KdV equation (\ref{KdV}). The following RH problem
characterizes $u(x,t,\e)$ at any time $t>0$.
\subsubsection*{RH problem for $M$:}
\begin{itemize}
\item[(a)] $M:\mathbb C\backslash \mathbb{R}\to \mathbb C^{2\times 2}$ is analytic. \item[(b)] $M$ has continuous boundary values $M_+(\lb)$ and $M_-(\lb)$ when approaching $\lambda\in\mathbb R\setminus\{0\}$ from above and below, and
\begin{align*}&M_+(\lb)=M_-(\lb){\small \begin{pmatrix}1&r(\lb;\e)
e^{2i\alpha(\lambda;x,t)/\e}\\
-\bar r(\lb;\e)e^{-2i\alpha(\lb;x,t)/\e}&1-|r(\lb;\e)|^2
\end{pmatrix}},&\mbox{ for $\lb<0$,}\\
&M_+(\lb)=M_-(\lb)\sigma_1,\quad
\sigma_1=\begin{pmatrix}0&1\\1&0\end{pmatrix},&\mbox{ for $\lb>0$},
\end{align*}
with $\alpha$ given by
\begin{equation}\label{def alpha1}\alpha(\lambda;x,t)=4t(-\lambda)^{3/2}+x(-\lambda)^{1/2}.\end{equation}
The branches of $(-\lambda)^{3/2}$ and $(-\lambda)^{1/2}$ are
analytic in $\mathbb C\setminus [0,+\infty)$ and positive for
$\lambda<0$.
\item[(c)] As $\lambda\to\infty$,
\begin{equation}\label{RHP M:c}M(\lb)=\left(I+\bigO(\lambda^{-1})\right)
\begin{pmatrix}1&1\\&\\i\sqrt{-\lb}&-i\sqrt{-\lb}\end{pmatrix}.
\end{equation}
\end{itemize}
The solution $M=M(\lambda;x,t,\e)$ depends on $x,t,\e$. If
$r(\lambda;\e)$ is the reflection coefficient from the left for the
Schr\"odinger equation $\e^2\frac{d^2}{dx^2}f+u_0(x)f=\lambda f$
with potential $u_0$, then it is known that
\begin{equation}\label{uM}
u(x,t,\e)=-2i\e\frac{\partial}{\partial
x}\lim_{\lambda\to\infty}\left(\sqrt{-\lambda}[M_{11}(\lambda;x,t,\e)-1]\right)
\end{equation}
is the solution to the KdV equation with initial data $u_0(x)$ at
time $t\geq 0$. Using certain smoothness and asymptotic (as $\e\to
0$) properties of the reflection coefficient, the RH problem for $M$
can be transformed to a RH problem with modified jump matrices. We
refer to \cite{CG} for the explicit construction of the function $S$
which satisfies the RH problem stated below, with jumps on a
deformed jump contour, see Figure \ref{figure: S}: lenses are opened
along an interval $(-1-\delta,u_c)$ for some small $\delta>0$. The
point $u_c$ is the Hopf solution $u(x,t)$ evaluated at the point
$x_c$ and time $t_c$ of gradient catastrophe.
\begin{figure}[t]
\begin{center}
    \setlength{\unitlength}{1.2mm}
    \begin{picture}(137.5,26)(22,11.5)
        \put(90,25){\thicklines\circle*{.8}}
        \put(45,25){\thicklines\circle*{.8}}
        \put(47,26){$-1-\delta$}
        \put(90,26){$0$}
         \put(74,26){$u_c$}
         \put(64,34,5){$\Sigma_1$}
         \put(64,13){$\Sigma_2$}
        \put(107,25){\thicklines\vector(1,0){.0001}}
        \put(90,25){\line(1,0){35}}

        \put(22,25){\line(1,0){23}}
        \put(35,25){\thicklines\vector(1,0){.0001}}
        \put(75,25){\line(1,0){15}}
        \put(84,25){\thicklines\vector(1,0){.0001}}
        \qbezier(45,25)(60,45)(75,25) \put(61,35){\thicklines\vector(1,0){.0001}}
        \qbezier(45,25)(60,5)(75,25) \put(61,15){\thicklines\vector(1,0){.0001}}
\end{picture}
\caption{The jump contour $\Sigma_S$ after the transformation
$M\mapsto S$}\label{figure: S}
\end{center}
\end{figure}
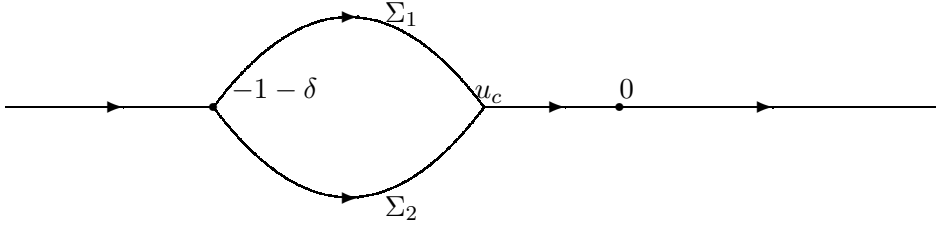
\subsubsection*{RH problem for $S$:}
\begin{itemize}
\item[(a)] $S$ is analytic in $\mathbb C\setminus \Sigma_S$,
\item[(b)] $S_+(\lambda)=S_-(\lambda)v_S(\lambda)$ for $\lambda\in\Sigma_S$,
with
\begin{equation}\label{vS}
v_S(\lambda)=\begin{cases}
\begin{pmatrix}1& i\kappa(\lb;\e)e^{\frac{2i}{\e}\phi(\lb)}\\
0&1
\end{pmatrix},&\mbox{ on $\Sigma_1$},\\[3ex]
 \begin{pmatrix}1&0\\
i\bar\kappa(\bar\lb;\e)e^{-\frac{2i}{\e}\phi(\lb;x)}&1
\end{pmatrix},&\mbox{ on $\Sigma_2=\overline{\Sigma_1}$,}\\[3ex]
\begin{pmatrix}e^{\frac{-2i}{\e}\phi_+(\lb;x)}
&i\kappa(\lambda;\e)\\
i\bar\kappa(\lambda;\e)&(1-|r(\lb)|^2)e^{\frac{2i}{\e}\phi_+(\lb;x)}
\end{pmatrix},&\mbox{ as $\lambda\in (u_c,0)$,}\\[3ex]
\sigma_1,&\mbox{\hspace{-3cm} as $\lambda\in(0,+\infty)$,}
\end{cases}
\end{equation}
and for $\lambda\in(-\infty,-1-\delta)$, we have
\begin{equation}
v_S(\lambda)=I+\bigO(e^{-\frac{c}{\e}(|\lambda|+1)}),\qquad\mbox{ as
$\e\to 0$, $c>0$,}
\end{equation}
uniformly in $\lambda$ for $x,t$ sufficiently close to $x_c,t_c$.
\item[(c)] $S(\lambda)=\left(I+\bigO(\lambda^{-1})\right)\begin{pmatrix}1&1\\
i\sqrt{-\lambda}&-i\sqrt{-\lambda}\end{pmatrix}$
as $\lambda\to\infty$.
\end{itemize}
$S$ can be expressed explicitly in terms of $M$ \cite{CG}, and it
follows from this explicit expression that
\begin{equation}\label{uS}
u(x,t,\e)=u_c-2i\e \frac{\partial}{\partial x}S_{1,11}(x,t,\e),
\end{equation}
where
\begin{equation}
S_{11}(\lb;x,t,\e)=1+\frac{S_{1,11}(x,t,\e)}{\sqrt{-\lb}}+\bigO(\lb^{-1}), \qquad \mbox{ as $\lb\to\infty$.}
\end{equation}
The function $\kappa$ can be expressed in terms of the (analytic
continuation of the) reflection coefficient and satisfies the
important asymptotic property
\begin{equation}
\kappa(\lambda;\e)=1+\bigO(\e),\qquad\mbox{ for
$\lambda\in\Sigma_1\cup[u,0]$},\qquad\mbox{ as $\e\to 0$.}
\end{equation}Furthermore $\phi$ depends explicitly on the
initial data:
\begin{multline}
\label{phi}
\phi(\lambda;x,t)=\sqrt{u_c-\lambda}(x-x_c-6u_c(t-t_c))+4(u_c-\lambda)^{3/2}(t-t_c)\\+\int_{\lambda}^{u_c}(f_L'(\xi)+6t_c)\sqrt{\xi-\lambda}d\xi.
\end{multline}
If
\begin{equation}\label{non-generic2}
f_L^{(2)}(u_c)=f_L^{(3)}(u_c)=f_L^{(4)}(u_c)=\ldots =
f_L^{(m)}(u_c)=0, \qquad f_L^{(m+1)}(u_c)\neq 0,
\end{equation}
repeated integration by parts gives (since $f_L'(u_c)+6t_c=0$)
\begin{multline}
\label{phi2}
\phi(\lambda;x,t)=\sqrt{u_c-\lambda}(x-x_c-6u_c(t-t_c))+4(u_c-\lambda)^{3/2}(t-t_c)\\
+\frac{2^m}{(2m+1)!!}\int_{\lambda}^{u_c}f_L^{(m+1)}(\xi)(\xi-\lambda)^{\frac{2m+1}{2}}d\xi.
\end{multline}
For any fixed neighborhood $\mathcal U$ of $u_c$, it was also proved
in \cite{CG} that there exists $\delta>0$ such that
\begin{equation}
v_S(\lambda)=\begin{cases} I+\bigO(e^{-\frac{c}{\e}}),&\mbox{ for
$\lambda\in\Sigma_S\setminus(\mathcal U\cup (u_c,+\infty))$,}\\
i\sigma_1+\bigO(\e),&\mbox{ for $\lambda\in(u_c,0)\setminus\mathcal
U$,}
\end{cases}
\qquad \mbox{ as $\e\to 0$,}
\end{equation}
if $|x-x_c|<\delta$ and $|t-t_c|<\delta$.

\subsection{Construction of the global parametrix}
If we ignore the jump matrices that are small as $\e\to 0$ and the
jumps in a fixed sufficiently small neighborhood $\mathcal U$ of
$u_c$, we obtain the following RH problem:
\subsubsection*{RH problem for $P^{(\infty)}$:}
\begin{itemize}
\item[(a)] $P^{(\infty)}:\mathbb C\setminus [u_c, +\infty) \to
\mathbb C^{2\times 2}$ is analytic, \item[(b)] $P^{(\infty)}$
satisfies the jump conditions
\begin{align}
&P_+^{(\infty)}=P_-^{(\infty)}\sigma_1, &\mbox{ on $(0, +\infty)$},\\
&P_+^{(\infty)}= iP_-^{(\infty)}\sigma_1, &\mbox{ on
$(u_c,0)$},\label{RHP Pinfty b}
\end{align}
\item[(c)]$P^{(\infty)}$ has the following behavior as $\lambda\to\infty$,
\begin{equation}P^{(\infty)}(\lambda)=(I+\bigO(\lambda^{-1}))
\begin{pmatrix}1&1\\ i(-\lambda)^{1/2} & -i(-\lambda)^{1/2}\end{pmatrix}
. \label{RHP Pinfty c}\end{equation}
\end{itemize}
This RH problem is solved by
\begin{equation}
\label{def Pinfty0}
P^{(\infty)}(\lambda)=(-\lambda)^{1/4}(u_c-\lambda)^{-\sigma_3
/4}\begin{pmatrix}1&1\\ i& -i
\end{pmatrix}.
\end{equation}

\subsection{Construction of the local parametrix}

We need to construct a local parametrix in a neighborhood $\mathcal
U$ of $u_c$. As $\e\to 0$, we have $\kappa(\lambda)=1+\bigO(\e)$,
and we will construct a function $P$, defined in $\mathcal U$, which
satisfies the same jump relations as $S$, but in the limiting case
where $\kappa$ is set to $1$.

\subsubsection*{RH problem for $P$:}
\begin{itemize}
\item[(a)]$P:\overline{\mathcal U}\setminus \Sigma_S\to\mathbb C^{2\times 2}$ is
analytic,
\item[(b)]$P$ satisfies the following jump condition on $\mathcal U \cap
\Sigma_S$,
\begin{equation}\label{RHP
P:b}P_+(\lambda)=P_-(\lambda)v_P(\lambda),
\end{equation}
with $v_P$ given by
\begin{equation}\label{vP}
v_{P}(\lb)=\begin{cases}
\begin{array}{ll}
\begin{pmatrix}1&ie^{\frac{2i}{\e}\phi(\lb;x,t)}\\
0&1
\end{pmatrix},&\mbox{ as $\lb\in\Sigma_1$},\\[3ex]
 \begin{pmatrix}1&0\\
ie^{-\frac{2i}{\e}\phi(\lb;x,t)}&1
\end{pmatrix},&\mbox{ as $\lb\in\Sigma_2$,}\\[3ex]
\begin{pmatrix}e^{-\frac{2i}{\e}\phi_+(\lb;x,t)}&i\\
i&0
\end{pmatrix},&\mbox{ as $\lambda\in (u_c,0)$,}%\\
\end{array}
\end{cases}
\end{equation}
\item[(c)] in the double scaling limit where  $\epsilon\to 0$ and simultaneously $x\to x_c$, $t\to t_c$ in such a way that
\begin{equation}\label{doublescaling}
\lim\frac{x- x_c-6u_c (t-t_c)}{\e^{\frac{2m+2}{2m+3}}
k^{1/2}}=\tau_0, \quad \lim\dfrac{-3(t-t_c)}{\e^{\frac{2m}{2m+3}}
k^{3/2}}= \tau_1,\quad \tau_0,\tau_1\in\mathbb R,
\end{equation}
with $k$ given by (\ref{k}), we have the matching
\begin{equation}\label{RHP P:c}
P(\lambda)P^{(\infty)}(\lambda)^{-1}\to I, \qquad \mbox{ for
$\lambda\in \partial \mathcal U$.}
\end{equation}
\end{itemize}

We will use the RH solution $\Psi=\Psi^{(m)}$ studied in Section
\ref{section: RHP} to construct the local parametrix $P$. First we
transform the RH problem for $\Psi$ to a RH problem for $\Phi$ which
models the jumps needed for $P$ in an appropriate way.

\subsubsection{Modified model RH problem}
Define \begin{equation}\label{def Phi1} \Phi(\zeta;s,t_1)=
e^{\frac{-\pi
i}{4}\sigma_3}\Psi(\zeta;s,t_1,0,\ldots,0)e^{\theta(\zeta;s,t_1,0,\ldots,
0)\sigma_3}
\begin{pmatrix}0&-1\\1&0\end{pmatrix}e^{\frac{\pi
i}{4}\sigma_3}\end{equation} for $\Im\zeta >0$, and
\begin{equation}\label{def Phi2} \Phi(\zeta;s,t_1)=e^{-\frac{\pi
i}{4}\sigma_3}\Psi(\zeta;s,t_1,0,\ldots,0)e^{\theta(\zeta;s,t_1,0,\ldots,0)\sigma_3}e^{\frac{\pi
i}{4}\sigma_3}
\end{equation}
for $\Im\zeta <0$. We also write
\begin{equation}
\widetilde\theta(\zeta;s,t_1)=-\frac{4}{2m+3}(-\zeta)^{\frac{2m+3}{2}}-\frac{4}{3}t_1(-\zeta)^{3/2}+s(-\zeta)^{1/2},
\end{equation}
which is related to $\theta$ in the case where $t_2=\ldots=t_{m-1}=0$, but with its branch cut on
$(0,+\infty)$. One has the identities
\begin{equation}
\theta=i\widetilde\theta_+,\quad\mbox{on $(0,+\infty)$,}\qquad
\theta=i\widetilde\theta,\quad\mbox{on $\Gamma_2$,}\qquad
\theta=-i\widetilde\theta,\quad\mbox{on $\Gamma_4$.}
\end{equation}
 Then it is straightforward to verify that $\Phi$
solves the RH problem
\subsubsection*{RH problem for $\Phi$:}
\begin{itemize}
    \item[(a)] $\Phi$ is analytic for $\zeta\in\mathbb{C}
    \setminus\widehat\Gamma$, with $\widehat\Gamma=\Gamma_1\cup\Gamma_2\cap\Gamma_4$.
    \item[(b)] $\Phi$ satisfies the following jump relations on
    $\widehat\Gamma$,
    \begin{align}
        \label{RHP Phi: b1}
        &\Phi_+(\zeta)=\Phi_-(\zeta)\begin{pmatrix}
            e^{-2i\widetilde\theta_+(\zeta;s,t_1)} & i \\
            i & 0
        \end{pmatrix}
        ,& \mbox{for $\zeta\in\Gamma_1$,} \\[1ex]
        \label{RHP Phi: b2}
        &\Phi_+(\zeta)=\Phi_-(\zeta)\begin{pmatrix}
            1 & ie^{2i\widetilde\theta(\zeta;s,t_1)} \\
            0 & 1
        \end{pmatrix}
        ,& \mbox{for $\zeta\in\Gamma_2$.}
        \\[1ex]
        \label{RHP Phi: b3}
        &\Phi_+(\zeta)=\Phi_-(\zeta)
        \begin{pmatrix}
            1 & 0 \\
            ie^{-2i\widetilde\theta(\zeta;s,t_1)} & 1
        \end{pmatrix},& \mbox{for $\zeta\in\Gamma_4$.}
    \end{align}
    \item[(c)] $\Phi$ has the following behavior at infinity,    \begin{equation}\label{RHP Phi: c}
        \Phi(\zeta)=\frac{1}{\sqrt
        2}(-\zeta)^{-\frac{1}{4}\sigma_3}\begin{pmatrix}1&1\\-1&1\end{pmatrix}\left(I+ih\sigma_3(-\zeta)^{-1/2}
        +\bigO(\zeta^{-1})\right),
    \end{equation}
    with the branches positive for $\zeta<0$ and analytic off
    $[0,+\infty)$.
\end{itemize}

We search for a parametrix $P$ of the form
\begin{equation}\label{definition hatP}
P(\lambda)=E(\lambda;\epsilon)\Phi(\epsilon^{-\frac{2}{2m+3}}f(\lambda);
\epsilon^{-\frac{2m+2}{2m+3}}\tau_0(\lambda;x,t),\epsilon^{-\frac{2m}{2m+3}}\tau_1(\lambda;t)),
\end{equation}
where $E$, $f$, $\tau_0, \tau_{1}$ are analytic in $\mathcal U$. So
we evaluate $\Phi(\zeta;s,t_1)$ at the values
\begin{align}
&\label{hat x}\zeta=\epsilon^{-\frac{2}{2m+3}}f(\lambda),&&
s=\epsilon^{-\frac{2m+2}{2m+3}}\tau_0(\lambda;x,t),\\
&t_1=\epsilon^{-\frac{2m}{2m+3}}\tau_1(\lambda;t),
\end{align}
and we will construct $f$, $\tau_0$, and $\tau_1$ in such a way that
\begin{equation}\label{condition f g 1}
\widetilde\theta(\e^{-\frac{2}{2m+3}}f(\lambda);\e^{-\frac{2m+2}{2m+3}}\tau_0(\lambda;x,t),\e^{-\frac{2m}{2m+3}}
\tau_1(\lambda;t))=\frac{1}{\e} \phi(\lambda;x,t).
\end{equation} This condition is satisfied if we
define $f$ by
\begin{equation}\label{def f}
-\frac{4}{2m+3}(-f(\lambda))^{\frac{2m+3}{2}}=\frac{2^m}{(2m+1)!!}
\int_{\lambda}^{u_c}f_L^{(m+1)}(\xi)(\xi-\lambda)^{\frac{2m+1}{2}}d\xi,
\end{equation}
$\tau_1$ by
\begin{equation}\label{def t1}
-\frac{4}{3}\tau_1(\lambda;t)(-f(\lambda))^{\frac{3}{2}}=4(t-t_c)(u_c-\lambda)^{3/2},
\end{equation}
and $\tau_0$ by
\begin{equation}\label{def t0}
\tau_0(\lambda;x,t)(-f(\lambda))^{\frac{1}{2}}=\sqrt{u_c-\lambda}(x-x_c-6u_c(t-t_c)).
\end{equation}
Indeed, summing (\ref{def f})-(\ref{def t0}) gives (\ref{condition f
g 1}) by (\ref{phi2}).
 This defines $f,\tau_0,\tau_1$ analytically near
$u_c$, and we have
\begin{align}
&\label{f0}f(u_c)=0,\qquad f'(u_c)=\left(-\frac{2^{m-1}}{(2m+1)!!}f_L^{(m+1)}(u_c)\right)^{\frac{2}{2m+3}}=k>0,\\
&\label{tau1}\tau_1(u_c)=-\frac{3(t-t_c)}{k^{3/2}},\\
&\label{tau0}\tau_0(u_c)=\frac{x-x_c-6u_c(t-t_c)}{k^{1/2}}.
\end{align}
Since $f$ is a conformal mapping from a neighborhood of $u_c$ to a
neighborhood of $0$, we can choose the lenses of the jump contour
for $S$ in such a way that $f(\Sigma_S\cap\mathcal
U)\subset\widehat\Gamma$. Then for any analytic function $E$ near
$u_c$, $P$ satisfies the required jump conditions on
$\Sigma_S\cap\mathcal U$ (see (\ref{vP})), but we also need the
matching (\ref{RHP P:c}), which has to be valid in the double
scaling limit where $\e\to 0$, $x\to x_c$, $t\to t_c$ in such a way
that (\ref{doublescaling}) holds, or in other words
\[
\lim\e^{-\frac{2m+2}{2m+3}}\tau_0(u_c;x,t)=\tau_0, \qquad
\lim\e^{-\frac{2m}{2m+3}}\tau_1(u_c;t)= \tau_1.
\]
If $\mathcal U$ is sufficiently small, $(\e^{-\frac{2m+2}{2m+3}}\tau_0(\lb;x,t), \e^{-\frac{2m}{2m+3}}\tau_1(\lb;t))$ will lie in a small complex neighborhood of $(\tau_0,\tau_1)$ for $\lb\in\partial \mathcal U$. By (\ref{RHP P:c}), we have
\begin{multline}
P(\lambda)P^{(\infty)}(\lambda)^{-1}=\frac{1}{\sqrt
        2}E(\lambda)(-\e^{-\frac{2}{2m+3}}f(\lambda))^{-\frac{1}{4}\sigma_3}\begin{pmatrix}1&1\\-1&1\end{pmatrix}\\
        \left(I+ih\sigma_3\e^{\frac{1}{2m+3}}
        (-f(\lambda))^{-1/2}
       +\bigO(\e^{\frac{2}{2m+3}})\right)
       P^{(\infty)}(\lambda)^{-1},
\end{multline}
as $\e\to 0$, where $h=h(\e^{-\frac{2m+2}{2m+3}}\tau_0(\lb;x,t), \e^{-\frac{2m}{2m+3}}\tau_1(\lb;t))$.
If we define
\begin{equation}
E(\lambda)=\frac{1}{\sqrt
        2}P^{(\infty)}(\lambda)\begin{pmatrix}1&-1\\1&1\end{pmatrix}(-\e^{-\frac{2}{2m+3}}f(\lambda))^{\frac{1}{4}\sigma_3},
\end{equation} it is easily verified that
$E$ is analytic in $\mathcal U$ and that we have
\begin{multline}
P(\lambda)P^{(\infty)}(\lambda)^{-1}=P^{(\infty)}(\lambda)\\
\left(I+ih\sigma_3\e^{\frac{1}{2m+3}}
        (-f(\lambda))^{-1/2}
         +\bigO(\e^{\frac{2}{2m+3}})\right)
         P^{(\infty)}(\lambda)^{-1},
\end{multline}
in the double scaling limit,
for $\lambda\in\partial \mathcal U$.

\subsection{Final RH problem}
We define $R$ in such a way that it has jumps that are uniformly
$I+\bigO(\e^{\frac{1}{2m+3}})$ in the double scaling limit: we let
 \begin{equation}
\label{def R imp}
R(\lambda; x,t,\e)=\begin{cases}
S(\lambda;x,t,\e)P^{(\infty)}(\lambda)^{-1}, &\mbox{ as $\lambda\in\mathbb C\setminus \overline{\mathcal U},$}\\
S(\lambda;x,t,\e)P(\lambda;x,t,\e)^{-1}, &\mbox{ as $\lambda\in \mathcal U.$}
\end{cases}
\end{equation}

Then, using the fact that
\begin{equation}\label{asymptotics vSP}
v_S(\lambda;x,t,\e)v_P^{-1}(\lambda;x,t,\e)=I+\bigO(\e), \qquad\mbox{ uniformly for
$\lambda\in\mathcal U\cap \Sigma_S$ as $\e\to 0$},
\end{equation}
one can verify that $R$ solves a RH problem of the following form.

\subsubsection*{RH problem for $R$:}
\begin{itemize}
\item[(a)] $R$ is analytic in $\mathbb C\setminus(\Sigma_S\cup \partial\mathcal U)$.
\item[(b)] $R$ has the jump condition $R_+(\lambda;x,t,\e)=R_-(\lambda;x,t,\e)v_R(\lambda;x,t,\e)$ for $\lambda\in \Sigma_S\cup \partial\mathcal U$,
 where
 \begin{align}
 &v_R(\lambda;x,t,\e)=I+\bigO(e^{-\frac{c}{\e}}), &\mbox{ for $\lambda\in\Sigma_S\setminus \overline{\mathcal U}$,}\\
 &v_R(\lambda;x,t,\e)=I+\bigO(\e), &\mbox{ for $\lambda\in\Sigma_S\cap \mathcal U$,}\\
  &v_R(\lambda;x,t,\e)=I+\bigO(\e^{\frac{1}{2m+3}}), &\mbox{ for $\lambda\in\partial \mathcal U$,}
 \end{align}
 in the double scaling limit where $\e\to 0$, $x\to x_c$, $t\to t_c$ and simultaneously $\e^{-\frac{2m+2}{2m+3}}\tau_0(u_c;x,t)\to
 \tau_0$ and $\e^{-\frac{2m}{2m+3}}\tau_1(u_c;t)\to
 \tau_1$.
\item[(c)] As $\lambda\to\infty$, we have
\begin{equation}\label{RHP R:c}R(\lambda;x,t,\e)=I+\frac{R_1(x,t,\e)}{\lambda}+\bigO(\lambda^{-2}).\end{equation}\end{itemize}
On $\partial\mathcal U$ with clockwise orientation, the jump matrix
has the form
\begin{equation}\label{vRexpansion}
v_R(\lambda;x,t,\e)=I+v_1(\lambda;x,t)\e^{\frac{1}{2m+3}}+\bigO(\e^{\frac{2}{2m+3}}),
\end{equation}
with
\begin{equation}
v_1(\lambda;x,t)=ih\cdot (-f(\lambda))^{-1/2} P^{(\infty)}(\lambda)\sigma_3
         P^{(\infty)}(\lambda)^{-1}.
\end{equation}
This is a meromorphic function in $\mathcal U$ with a simple pole at $u_c$, the residue is given by
\begin{equation}
\Res
(v_1;u_c)=-h(\e^{-\frac{2m+2}{2m+3}}\tau_0(u_c;x,t),\e^{-\frac{2m}{2m+3}}\tau_1(u_c;t))k^{-1/2}\begin{pmatrix}0&1\\0&0\end{pmatrix}.
\end{equation}
Then as in \cite[Section 4]{CG} one can conclude that
\begin{equation}
R_1(x,t,\e)=\e^{\frac{1}{2m+3}}\Res (v_1;u_c)+\bigO(\e^{\frac{2}{2m+3}}),
\end{equation}
and by (\ref{uS}) and (\ref{tau0}) this leads to
\begin{eqnarray*}
u(x,t,\e)&=&u_c-2\epsilon\frac{\partial}{\partial x}
R_{1,12}(x,t,\e)\\
&=&u_c+2k^{-1/2}\e^{\frac{2}{2m+3}}\frac{\partial\tau_0(u_c;x,t)}{\partial x}
q(\e^{-\frac{2m+2}{2m+3}}\tau_0(u_c;x,t),\e^{-\frac{2m}{2m+3}}\tau_1(u_c;t),0,\ldots, 0)+\bigO(\e^{\frac{4}{2m+3}})\nonumber\\
&=&u_c-2k^{-1}\epsilon^{\frac{2}{2m+3}}
q(\e^{-\frac{2m+2}{2m+3}}\tau_0(u_c;x,t),\e^{-\frac{2m}{2m+3}}\tau_1(u_c;t),0,\ldots,
0)+\bigO(\e^{\frac{4}{2m+3}}),
\end{eqnarray*}
which proves Theorem \ref{theorem: KdV}.

\section*{Acknowledgements}
The author acknowledges support by the Belgian Interuniversity
Attraction Pole P06/02.

\end{document}